\documentclass[a4paper,UKenglish]{lipics-v2018}

\usepackage{xspace}
\usepackage{microtype}
\usepackage{amsmath}
\usepackage{amsthm}
\usepackage{amssymb}
\usepackage{hyperref}
\usepackage{stmaryrd}
\usepackage{bbold}
\usepackage{tikz}
\usetikzlibrary{arrows,decorations.pathmorphing,backgrounds,positioning,fit,calc,automata}
\usetikzlibrary{trees}
\usetikzlibrary{shapes}

\usepackage[textwidth=2cm,textsize=small]{todonotes}

\newcommand*{\eg}{e.g.\@\xspace}
\newcommand*{\ie}{i.e.\@\xspace}

\newcommand{\LRS}{\mathbf{LRS}}
\renewcommand{\aa}{{\bf a}}
\newcommand{\uu}{{\bf u}}
\newcommand{\vv}{{\bf v}}
\newcommand{\sequence}[1]{\langle #1_n \rangle_{n \in \N}}
\newcommand{\shuffleop}[1]{\text{shuffle}(#1)}

\newcommand{\Rat}{\mathbf{Rat}}
\newcommand{\PRat}{\mathbf{PolyRat}}
\newcommand{\Fin}{\mathbf{Fin}}
\newcommand{\Arith}{\mathbf{Arith}}
\newcommand{\Geo}{\mathbf{Geo}}

\newcommand{\shift}{\text{shift}}
\newcommand{\shuffle}{\text{shuffle}}

\newcommand{\CCRA}{\mathbf{CCRA}}
\newcommand{\LCRA}{\mathbf{LCRA}}
\newcommand{\WA}{\mathbf{WA}}

\newcommand{\DetWA}{\mathbf{DetWA}}
\newcommand{\oneWA}{\mathbf{1WA}}
\newcommand{\kWA}{\mathbf{kWA}}
\newcommand{\kplusoneWA}{\mathbf{(k+1)WA}}
\newcommand{\FinWA}{\mathbf{FinWA}}
\newcommand{\PolyWA}{\mathbf{PolyWA}}

\newcommand{\A}{\mathcal{A}}

\newcommand{\C}{\mathcal{C}}

\newcommand{\N}{\mathbb{N}}

\newcommand{\Q}{\mathbb{Q}}
\newcommand{\R}{\mathbb{R}}

\newcommand{\set}[1]{\left\{ #1 \right\}}

\newcommand{\sem}[1]{\llbracket #1\rrbracket}
\newcommand{\Runs}{Runs_{\A}(n)}

\newcommand{\X}{\mathcal{X}}
\newcommand{\Expr}{Expr(\X)}
\newcommand{\Subs}{Subs(\X)}

\title{A Robust Class of Linear Recurrence Sequences}

\author{Corentin Barloy}{{\'E}cole Normale Sup{\'e}rieure de Paris, France}{}{}{}

\author{Nathana{\"e}l Fijalkow}{CNRS, LaBRI, Bordeaux, France, and the Alan Turing Institute of data science, London, United Kingdom}{}{}{}

\author{Nathan Lhote}{University of Warsaw, Poland}{}{}{}

\author{Filip Mazowiecki}{LaBRI, Universit{\'e} de Bordeaux, France}{}{}{}

\authorrunning{C. Barloy et al.}

\Copyright{C. Barloy and N. Fijalkow and N. Lhote and F. Mazowiecki}



\category{}

\relatedversion{}

\supplement{}

\funding{}

\EventEditors{John Q. Open and Joan R. Access}
\EventNoEds{2}
\EventLongTitle{42nd Conference on Very Important Topics (CVIT 2016)}
\EventShortTitle{CVIT 2016}
\EventAcronym{CVIT}
\EventYear{2016}
\EventDate{December 24--27, 2016}
\EventLocation{Little Whinging, United Kingdom}
\EventLogo{}
\SeriesVolume{42}
\ArticleNo{23}
\nolinenumbers 

\begin{document}

\maketitle

\begin{abstract}
We introduce a subclass of linear recurrence sequences which we call poly-rational sequences
because they are denoted by rational expressions closed under sum and product.
We show that this class is robust by giving several characterisations: 
polynomially ambiguous weighted automata, copyless cost-register automata, 
rational formal series, and linear recurrence sequences whose eigenvalues are roots of rational numbers.
\end{abstract}

\section{Introduction}
\label{sec:introduction}
The study of sequences of numbers originated in mathematics and has deep connections with many fields.
A prominent class of sequences is \emph{linear recurrence sequences}, such as the Fibonacci sequence
\[
0,1,1,2,3,5,8,13,\ldots
\]
Despite the simplicity of linear recurrence sequences many problems related to them remain open, 
and are the object of active research. 
In theoretical computer science the two main questions are:
\begin{itemize}
	\item How to finitely represent sequences?
	\item How to algorithmically analyse properties of sequences?
\end{itemize}

In this paper we focus on problems related to the first question. 
The question of representation has led to important insights in the structure of linear recurrence sequences
by giving several equivalent characterisations, some of which we briefly review here.
We refer to Section~\ref{sec:rat} and the next sections for technical definitions.

\subparagraph*{Linear recurrence sequences}
A sequence of real numbers $\uu = \sequence{u} = \langle u_0, u_1, u_2, \ldots \rangle$ is a linear recurrence system (LRS) if there exist real numbers $a_1, \ldots, a_k$ such that 
for all $n \ge 0$
\begin{align}\label{eq:lrs}
u_{n+k} \; = \; a_1 u_{n+k-1} + \ldots + a_k u_{n}.
\end{align}
In this paper we will consider only sequences of rational numbers, therefore, we additionally assume that $a_i$ are rational numbers. 
The smallest $k$ for which $\uu$ satisfies an equation of the form~\eqref{eq:lrs} is called the order of $\uu$. 
The Fibonacci sequence $\sequence{F}$ is an LRS of order~$2$ satisfying the recurrence $F_{n+2} = F_{n+1} + F_n$.

\subparagraph*{Rational expressions}
Studying the closure properties of linear recurrence sequences yields the following result,
an instance of the Kleene-Sch\"{u}tzenberger theorem~\cite{Schutzenberger61b}:
linear recurrence sequences form the smallest class of sequences containing the sequences $\langle a,0,0,\ldots \rangle$ for a rational number $a$
and closed under sum, Cauchy product, and Kleene star.

\subparagraph*{Weighted automata}
The model of weighted automata is a well studied quantitative extension of classical automata.
In general a weighted automaton recognises a function $f : \Sigma^* \to \R$,
hence when considering a unary alphabet this becomes $f : \set{a}^* \to \R$, 
and identifying $\set{a}^*$ with $\N$ we can see $f$ as a sequence of numbers.
Whenever we write about sequences recognised by models like weighted automata, we implicitly assume that these are over a unary alphabet.

\subparagraph*{Cost-register automata}
Several characterisations of weighted automata have been introduced~\cite{DrosteG07,KreutzerR13,AlurDDRY13}.
We will be interested in the model of \emph{cost-register automata} (CRA). 
These are deterministic models with registers whose contents are blindly updated (i.e., without transitions like zero tests). 
It was shown that considering linear updates yields a model equivalent to weighted automata.

\vskip1em
We summarise in one theorem the equivalences above, which is the starting point of our work.
Technical definitions are given in the paper.
\begin{theorem}[Folklore, see for instance~\cite{Bousquet-Melou05,Schutzenberger61b,DrosteHWA09}]\label{th:classical}
The following classes of sequences are effectively equivalent.
\begin{itemize}
	\item Linear recurrence sequences,
	\item Sequences recognised by weighted automata,
	\item Sequences recognised by linear cost-register automata,
	\item Sequences denoted by rational expressions,
	\item Sequences whose formal series are rational, \textit{i.e.} of the form $\frac{P}{Q}$ where $P,Q$ are polynomials.
\end{itemize}
\end{theorem}

\paragraph*{Algorithmic analysis of linear recurrence sequences}
The questions regarding algorithmic analysis are far from being answered. 
A very simple and natural problem, the Skolem problem, is still unsolved~\cite{tao2008structure,OuaknineW15}:
given a linear recurrence sequence, does it contain a zero?
Recent breakthrough results sharpened our understanding of the Skolem problem~\cite{OuaknineW14,OuaknineW14a},
but one of the outcomes is that the general problem for the whole class of linear recurrence sequences is beyond our reach at the moment,
since it would impact notoriously difficult problems from number theory.
We refer the reader to the recent survey about what is known to be decidable for linear recurrence sequences~\cite{OuaknineW15}. 

\paragraph*{Our contributions}
Since the full class of linear recurrence sequences is too hard to be algorithmically analysed (we only mentioned the Skolem problem but many related problems are also difficult), let us revise our ambitions, go back to the drawing board, and study tractable subclasses.

In this paper we introduce \emph{poly-rational sequences} which is a strict fragment of linear recurrence sequences. 
We give several equivalent characterisations of this class following the equivalence results stated in Theorem~\ref{th:classical}.
Our results are summarised in the following theorem.
\begin{theorem}\label{th:main}
The following classes of sequences are effectively equivalent.
\begin{itemize}
	\item Sequences denoted by poly-rational expressions (Section~\ref{sec:rat}),
	\item Sequences recognised by polynomially ambiguous weighted automata (Section~\ref{sec:poly}),
	\item Sequences recognised by copyless cost-register automata (Section~\ref{sec:ccra}),
	\item Sequences whose formal series are of the form $\frac{P}{Q}$ where $P,Q$ are polynomials 
	and the roots of $Q$ are roots of rational numbers (Section~\ref{sec:lrs}),
	\item Linear recurrence sequences whose eigenvalues are roots of rational numbers (Section~\ref{sec:lrs}).
\end{itemize}
\end{theorem}

We do not discuss the efficiency of reductions proving the equivalences. 
Our constructions are elementary, and in most cases they yield blow ups in the size of representation.

We note that the Skolem problem and its variants are known to be decidable, and NP-hard, for the subclass of poly-rational sequences.
The decidability easily follows from the fact that our class is subsumed by other classes for which such results were obtained (see e.g.~\cite{RebihaMM14b}, for the case where all eigenvalues are roots of algebraic real numbers).
The Skolem problem is known to be NP-hard already for the class of LRS whose eigenvalues are roots of unity~\cite{ABV17}. 
This implies that the Skolem problem for the class of poly-rational sequences is also NP-hard, 
which is the best known lower bound even for the full class of linear recurrence sequences.

\paragraph*{Related works}
The intractability of the Skolem problem for linear recurrence sequences also impacts the other equivalent models,
leading to the study of several restrictions.
A classical approach to tame weighted automata is to bound the ambiguity of weighted automata, \ie 
bounding the number of accepting runs with a function depending on the length of the word. 
Many positive results have been obtained in the past years following this approach~\cite{KlimannLMP04,KirstenL09,FijalkowRW17}.

Another restriction studied in the model of cost-register automata is the \emph{copyless} restriction: 
registers are not allowed to be copied more than once.
It was conjectured that the copyless restriction would result in good decidability properties~\cite{AlurDDRY13}, 
but this has been recently falsified~\cite{AlmagorCMP18}.

\section{Linear recurrence sequences and rational expressions}
\label{sec:rat}
We let $\uu = \sequence{u} = \langle u_0,u_1,u_2\ldots \rangle$ denote a sequence of rational numbers.

\paragraph*{Linear recurrence sequences}
We will assume that an LRS $\uu$ is given by the numbers $a_1, \ldots, a_k$ and the values of the first $k$ elements: 
$u_0,\ldots, u_{k-1}$. The recurrence~\eqref{eq:lrs} induces the sequence $\uu$.
We let $\LRS$ denote the class of LRS.
Given an LRS we define its characteristic polynomial as
\begin{align*}
Q(x) \; = \; x^k - a_{1}x^{k-1} - \ldots - a_{k-1}x - a_k.
\end{align*}
The roots of the characteristic polynomial are called the \emph{eigenvalues} of the LRS.

\paragraph*{Formal series}
Formal series are a different representation for sequences.
The sequence $\sequence{u}$ induces the formal series $S(x) = \sum_{n \in \N} u_n x^n$, 
with the interpretation that the coefficient of $x^n$ is the value of the $n$-th element in the sequence.
Note that a polynomial represents a sequence with a finite support. 

\begin{example}\label{ex:fib}
A standard example of an LRS is the Fibonacci sequence $\sequence{F}$ defined by the recurrence $F_{n+2} = F_{n+1} + F_n$ and initial values $F_0 = 0, F_1 = 1$. 
Its characteristic polynomial is $p(x) = x^2 - x - 1$, whose roots are $\frac{1+\sqrt{5}}{2}$ and $\frac{1 - \sqrt{5}}{2}$. 
The corresponding formal series is $S(x) = \sum_{n=0}^\infty F_nx^n$. 
Using the definition of $F$ we obtain $S(x) = x + xS(x) + x^2S(x)$ and thus $S(x) = \frac{x}{1 - x - x^2}$.
\end{example}

\paragraph*{Rational expressions}
We start by defining three classes of sequences.
\begin{itemize}
	\item $\Fin$: a sequence $\uu$ is in $\Fin$, or equivalently $\uu$ has finite support, if the set $\set{n \in \N : u_n \neq 0}$ is finite;
	\item $\Arith$: a sequence $\uu$ is in $\Arith$, or equivalently $\uu$ is arithmetic, if $u_0 = a$, $u_{n+1} = u_n + b$ for some rational numbers $a,b$;
	\item $\Geo$: a sequence $\uu$ is in $\Geo$, or equivalent $\uu$ is geometric, if $u_0 = a$, $u_{n+1} = \lambda \cdot u_n$, for some rational numbers $a,\lambda$.
\end{itemize}
We let $\Geo_\lambda$ denote the class of geometric sequences with a fixed parameter $\lambda$.

\vskip1em
We now define some classical operators. Here $\uu, \vv, \uu^1, \ldots, \uu^k$ are sequences.
\begin{itemize}
	\item \textbf{Sum}: $\uu + \vv$ is the component wise sum of sequences;
	\item \textbf{Cauchy product}: $\uu \cdot \vv = \langle \sum_{p+q=n} u_p \cdot v_q \rangle_{n \in \N}$;
	inducing $(\uu)^n$ defined by $(\uu)^0 = \langle 1,0,0,0,\ldots \rangle$ and $(\uu)^{n+1} = (\uu)^n \cdot \uu$, in particular $(\uu)^1 = \uu$;
	\item \textbf{Kleene star}: $\left(\uu \right)^* = \sum_{n \in \N} \left( \uu \right)^n$, it is only defined when $u_0 = 0$;
	\item \textbf{Hadamard product}: $\uu \times \vv$ is the component wise product of sequences;
	\item \textbf{Shift}: $\langle a,\uu \rangle = \langle a,u_0,u_1,\ldots \rangle$, defined for any rational number $a$;
	\item \textbf{Shuffle}: $\shuffleop{\uu^1,\uu^2,\ldots,\uu^k} = \langle u^1_0,u^2_0,\ldots,u^k_0,u^1_1,u^2_1,\ldots,u^k_1,u^1_2,\ldots \rangle$.
\end{itemize}

We write $\Rat[\C,\text{op}_1,\ldots,\text{op}_k]$ for the smallest class of sequences containing $\C$ and closed under the operators
$\text{op}_1,\ldots,\text{op}_k$.
Rational expressions in Theorem~\ref{th:classical} are classically defined as follows~\cite{Schutzenberger61b}:
\[
\Rat = \Rat[\Fin,+,\cdot,*].
\]
The class $\Rat$ contains all classes defined above, and is closed under all mentioned operators, \ie 
\[
\Rat = \Rat[\Fin \cup \Arith \cup \Geo,+,\cdot,*,\times,\shift,\shuffle].
\]

We now introduce a class of sequences denoted by a fragment of rational expressions,
whose study is the purpose of this article.
The class is called poly-rational sequences, because they are denoted by rational expressions using sum and product.

\begin{definition}[Poly-rational sequences]
\[
\PRat = \Rat[\Arith \cup \Geo, +, \times, \shift, \shuffle].
\]
\end{definition}
In other words $\PRat$ is the smallest class of sequences containing arithmetic and geometric sequences that is
closed under sum, Hadarmard product, shift, and shuffle. A trivial observation is that $\Fin \subseteq \PRat$ since using $\shift$ one can generate any sequence with finite support. One could try to simplify the definition of $\PRat$ replacing $\Arith \cup \Geo$ with $\Fin$. Unfortunately, the operators $+, \times, \shift, \shuffle$ are too restricted, and geometric and arithmetic sequences could not be generated. In fact, the class would collapse to $\Fin$. 

Since $\Rat$ contains $\Arith$ and $\Geo$ and is closed under Hadamard product, shift, and shuffle,
we have $\PRat \subseteq \Rat$. We will show that the inclusion is indeed strict.
As we will see in this paper, the class $\PRat$ has many equivalent and surprising characterisations.

\section{Characterisation with polynomially ambiguous weighted automata}
\label{sec:poly}
We refer to e.g.~\cite{DrosteHWA09} for an excellent introduction to weighted automata.
We consider weighted automata over the rational semiring $(\Q,+,\cdot)$, where $+$ and $\cdot$ are the standard sum and product.
For an alphabet $\Sigma$, weighted automata recognise functions assigning rational numbers to finite words,
\ie $f : \Sigma^* \to \Q$. 
In this paper we will consider only one-letter alphabets so the set of words is $\set{a}^* = \set{\varepsilon, a, a^2, \ldots}$,
which is identified with $\N$.
Therefore, weighted automata recognise functions $f : \N \to \Q$, \ie
weighted automata recognise sequences of rational numbers.

Formally, a weighted automaton is a tuple $\A = (Q, M, I, F)$, where $Q$ is a finite set of states, 
$M$ is a $Q \times Q$ matrix  over $\Q$ and $I, F$ are the initial and final vectors, respectively, of dimension $Q$ (for convenience we label the coordinates by elements of $Q$). 
The sequence recognised by the automaton $\A$ is $\sem{\A}$ defined by $\sem{\A}(n) = I^t M^n F$, 
where $I^t$ is the transpose of $I$. 

We give an equivalent definition of $\A$ in terms of accepting runs. 
We say that a state $q \in Q$ is an initial state if $I(q) \neq 0$ and that it is a final state if $F(q) \neq 0$. If $q$ is initial we say that its initial weight is $I(q)$, and if $q$ is final then its final weight is $F(q)$.
For two states $p,q \in Q$ we say that there is a transition from $p$ to $q$ if $M(p,q) \neq 0$. Such a transition is denoted $p \to q$ and its weights is $M(p,q)$.
A run $\rho$ is a sequence of consecutive transitions, and it is accepting if the first state is initial and the last state is final.
The value of an accepting run $\rho = q_0 \to q_1 \to \cdots \to q_n$ is 
\[
|\rho| = I(q_0) \cdot \left( \prod_{i = 0}^{n-1} M(q_i,q_{i+1}) \right) \cdot F(q_n).
\]
Let $\Runs$ denote the set of all accepting runs of length $n$. 
An alternative and equivalent definition of $\sem{\A}$ is 
\[
\sem{\A}(n) = \sum_{\rho \in \Runs} |\rho|.
\]

\begin{example}\label{ex:weighted}
Consider the automaton $\A = (Q, M, I, F)$ represented in Figure~\ref{fig:fib}.
We have $\sem{\A}(n) = F_n$, where $\sequence{F}$ is the Fibonacci sequence from Example~\ref{ex:fib}.
\begin{figure}[!ht]
\centering
\includegraphics[scale=.5]{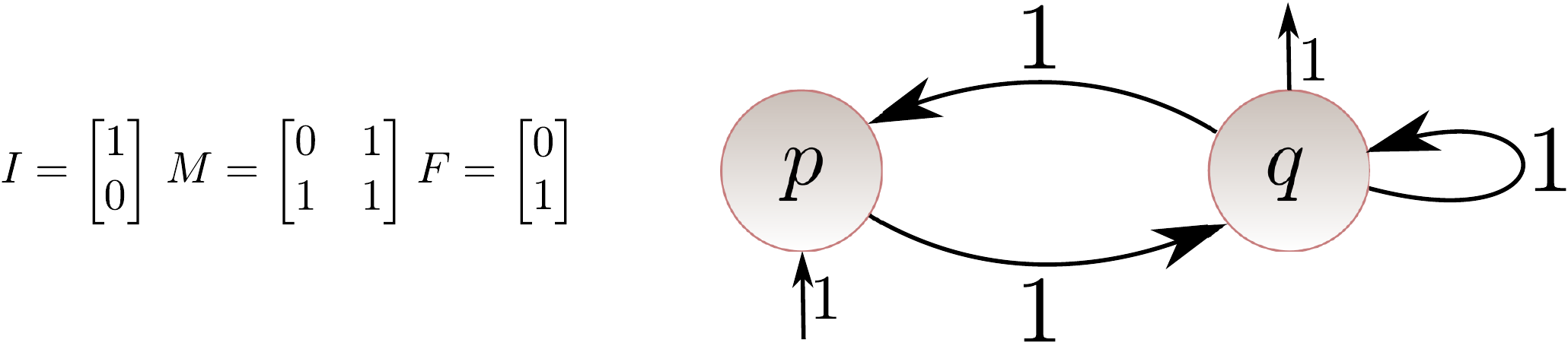}
\caption{A weighted automaton recognising the Fibonacci sequence.}
\label{fig:fib}
\end{figure}
\end{example}

The ambiguity of an automaton $\A$ is the function $a_\A : \N \to \N$ which associates to $n$ the number of accepting runs $|\Runs|$.
We consider the following classes:
\begin{itemize}
 \item $\DetWA$ -- the class of deterministic weighted automata;
 \item $\kWA$ for fixed $k \in \N$ -- the class of $k$-ambiguous weighted automata, \ie when $a_\A(n) \le k$ for all $n$;
 \item $\FinWA = \bigcup_{k \in \N}\kWA$ -- the class of finitely ambiguous weighted automata, \ie when there exists $k$ such that $a_\A(n) \le k$ for all $n$;
 \item $\PolyWA$ -- class of polynomially ambiguous automata, \ie when there exists a polynomial $P : \N \to \N$ such that $a_\A(n) \le P(n)$ for all $n$;
 \item $\WA$ -- the full class of weighted automata.
\end{itemize}

For example, the automaton in Example~\ref{ex:weighted} is not polynomially ambiguous because the number of accepting runs is exponential.
We will see that this is no accident by proving in Section~\ref{sec:lrs} that the Fibonacci sequence is not in $\PolyWA$.

We present our first characterisation of $\PRat$.

\begin{theorem}\label{th:poly}
$\PRat = \PolyWA$
\end{theorem}

\subsection*{Proof of Theorem~\ref{th:poly}}
This subsection is divided into two parts for both inclusions.

\subsection*{$\PRat \subseteq \PolyWA$}

Figure~\ref{fig:arithgeo} shows how to recognise the arithmetic and the geometric sequences.
\begin{figure}[!ht]
\centering
\includegraphics[scale=.5]{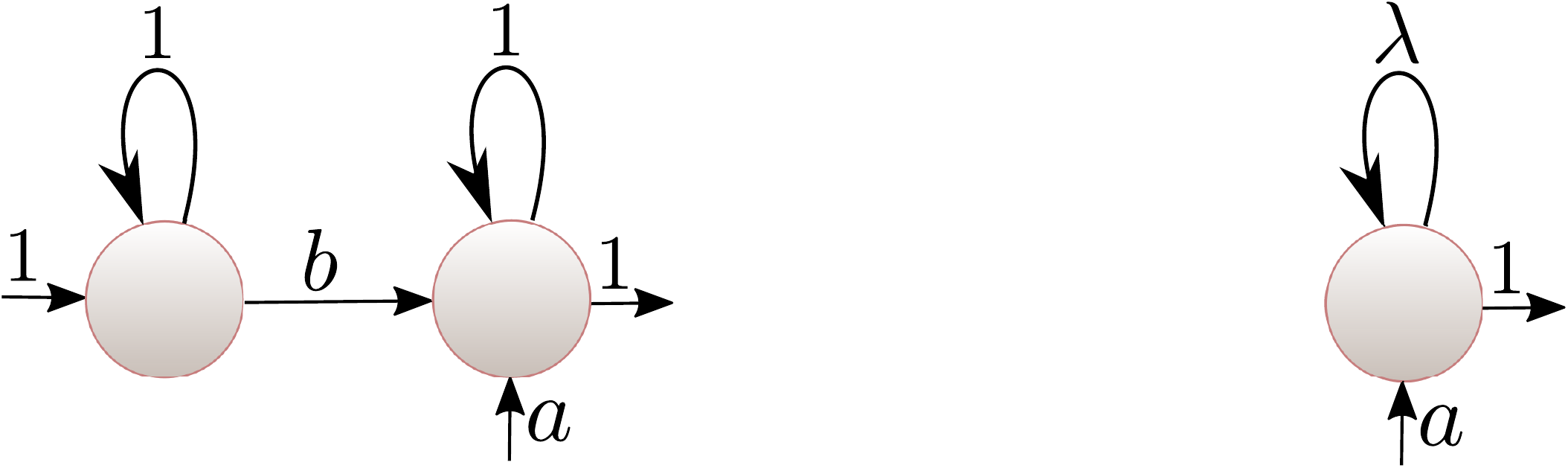}
\caption{The weighted automaton on the left recognises the arithmetic sequence with parameters $(a,b)$ and it is linearly ambiguous.
The weighted automaton on the right recognises the geometric sequence with parameters $a,\lambda$ and it is deterministic.}\label{fig:arithgeo}
\label{fig:arithmetic-geometric}
\end{figure}
For each finitely supported sequence a simple weighted automaton can be constructed.
It remains to prove that the class $\PolyWA$ is closed under the operators.
The sum and products correspond to union and product of automata,
it is readily verified that these standard constructions preserve the polynomial ambiguity. Below we deal with shift and shuffle operators.

Suppose we have a polynomially ambiguous automaton $\A$ for $\uu$ and we want to construct a new polynomially ambiguous automaton $\A'$ for $\langle a,\uu \rangle$. We start with the case when $a = 0$. Then $\A'$ has the same state as $\A$ plus one new state $q_0$, which is the only initial state in $\A'$. All transitions from $\A$ are inherited. There are additionally only outgoing transitions from $q_0$ to all states that are initial in $\A$; the weight of the transition is the initial weight of the corresponding state in $\A$. It is readily verified that $\A'$ recognises $\langle 0,\uu \rangle$ and that $\A'$ is polynomially ambiguous.
For $a \neq 0$ it suffices to add one more state that is both initial and final with initial weight $1$ and final weight $a$.

To deal with shuffle we start with the following preliminary construction. Fix some $k > 0$ and a polynomially ambiguous automaton $\A$ recognising $\uu$. We construct $\A[k]$ recognising $\uu' = \langle \underbrace{u_0, 0,\ldots,0}_{k}, \underbrace{u_1, 0,\ldots,0}_{k}, u_2, \ldots \rangle$, \ie elements $u_i$ are separated by $k-1$ elements with $0$. The idea to construct $\A'$ is that the set of states have an additional component $\set{0,\ldots, k-1}$, and they behave like $\A$ every $k$-th step; in the remaining steps they only wait.
Formally, the set of states of $\A[k]$ is $Q \times \set{0,\ldots,k-1}$, where $Q$ is the set of states of $\A$. The initial (final) states are $(q,0)$ such that $q$ is initial (final) in $\A$ with the same weight. For every transition $p \to q$ in $\A$ there is a transition $(p,0) \to (q,1)$ in $\A[k]$ with the same weight. The remaining transitions are $(q,i) \to (q,(i+1) \mod k)$ with weight $1$, defined for every $i > 0$ and every $q\in Q$. It is readily verified that $\A[k]$ recognises $\uu'$.

Let $\A_0,\ldots,\A_{k-1}$ be polynomially ambiguous automata recognising $\uu_0,\ldots,\uu_{k-1}$.
For every $\A_i$ let $\A_i[k]$ be an automaton as above, additionally shifted $i$ times with $0$'s. Then $\shuffleop{\uu_0,\ldots,\uu_{k-1}}$ is recognised by the disjoint union of $\A_i[k]$.

\subsection*{$\PolyWA \subseteq \PRat$}

The first step is to decompose polynomially ambiguous automata into a union of automata that we will call \emph{chained loops}.
We say that the states $p_0, p_1, \ldots p_{k-1} \in Q$ form a \emph{loop} if $M(p_i, p_j) \neq 0$ is equivalent to $j = i+1 \mod k$
and a \emph{path} if $M(p_i, p_j) \neq 0$ is equivalent to $j = i+1$ (in particular $p_{k-1}$ has no successor).
A chained loop of size $k$ is an automaton over the set states of $\set{q_0, \ldots, q_{k-1}} \cup P$ such that
\begin{itemize}
	\item $q_0$ is the unique initial state;
	\item $q_0, \ldots, q_{k-1}$ form a path;
	\item each $q_i$ is contained in at most one loop (the states in $P$ are used only as intermediate states in the loops);
	\item $q_{k-1}$ is the unique final state with $F(q_{k-1}) = 1$.
\end{itemize}
We define the concatenation of two chained loops $\A_1, \A_2$: this is the chained loop
obtained by constructing the union of the two automata
with the initial state being the initial state of $\A_1$, the final state being the final state of $\A_2$, and
rewiring the output of $\A_1$ to the initial state of $\A_2$, see \eg Figure~\ref{fig:chained}.

\begin{figure}[!ht]
\centering
\includegraphics[scale=.4]{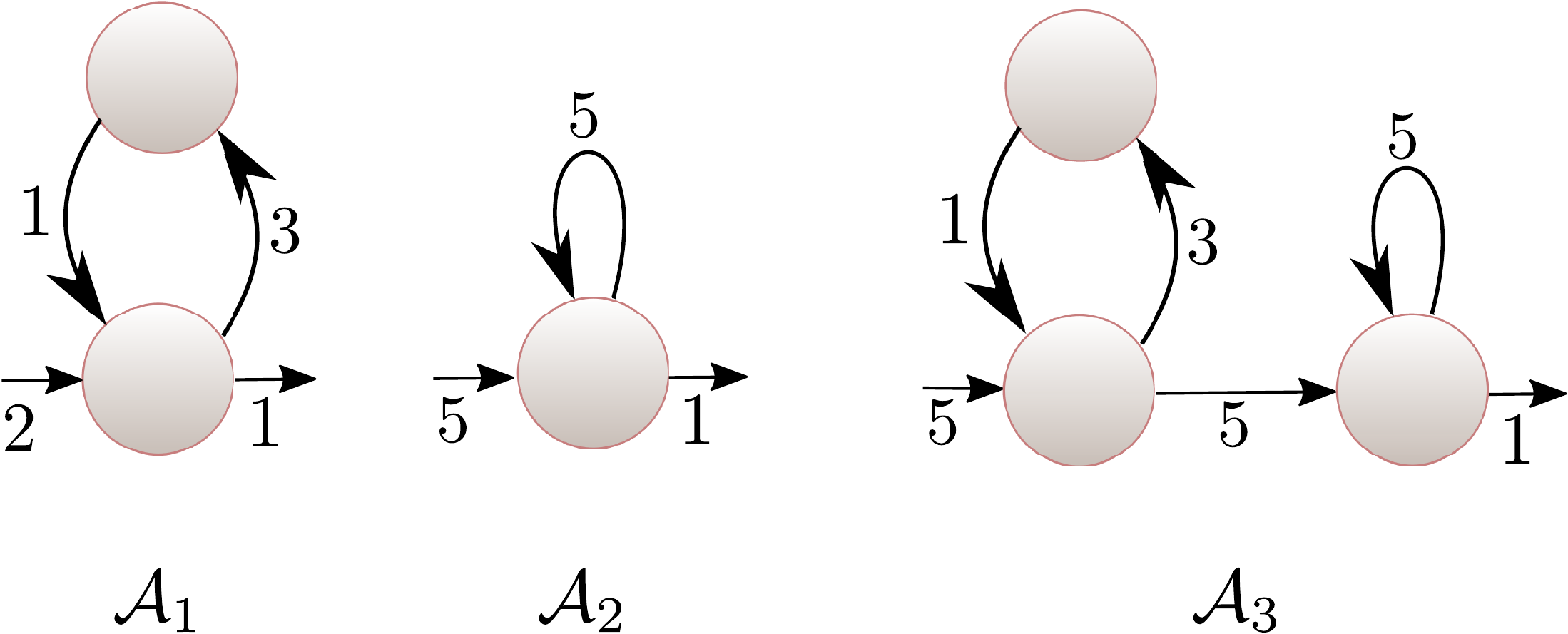}
\caption{Three example chained loops. The initial and final weights are depicted by ingoing and outgoing edges. The chained loop $\A_1$ recognises the sequence defined by $f_1(2n) = 2 \cdot 3^n$, $f_1(2n+1) = 0$ whose power series is $\frac{2}{1 - 3x^2}$. The chained loop $\A_2$ recognises the sequence $f_2(n) = 5^{n+1}$ whose power series is $\frac{5}{1-5x}$. The chained loop $\A_3$ is the concatenation of $\A_1$ and $\A_2$ and it recognises the sequence $f_3(n) = \sum_{i=1}^n f_1(i-1) \cdot f_2(n-i)$ whose power series is $\frac{10x}{(1-3x^2)(1-5x)}$.}\label{fig:chained}
\end{figure}

\begin{lemma}\label{lem:poly-chained}
Any polynomially ambiguous weighted automaton is equivalent to a union of chained loops.
\end{lemma}

\begin{proof}
Let $\A$ be a polynomially ambiguous weighted automaton. 
Without loss of generality $\A$ is trimmed, \ie every state occurs in at least one accepting run.

We first note that any state in $\A$ is contained in at most one loop.
Indeed, a state contained in two loops induces a sequence of words with exponential ambiguity.
This implies that a sequence $(q_0,q_1,\dots,q_k)$ with $q_i \neq q_j$ for $i \neq j$ induces at most one chained loop of which it is the path.
There are finitely many such sequences because $k < |\A|$.

We claim that $\A$ is equivalent to the union of all chained loops induced by such sequences.
Indeed, there is a bijection between the runs of $\A$ and the runs of all the chained loops, respecting the values of runs.
Consider a run $\rho$ of $\A$, where a state $q$ appears multiple times. 
Then between each occurence of $q$ this is the same run, because they are loops over $q$ and there can be only one loop containing $q$.
So $\rho = u v^k w$, where $v$ is the (only) loop containing $q$.
Repeating this for $u$ and $w$, we obtain a unique decomposition of $\rho$ into 
\[
q_0 \cdot \ell_0^{m_0} \cdot q_0 \to q_1 \cdot \ell_1^{m_1} \cdot q_1 \to q_2 \dots q_k \cdot \ell_n^{m_k} \cdot q_k,
\]
where $\ell_i$ is a loop over $q_i$ (we can have $m_i = 0$) and $q_i \neq q_j$ for $i \neq j$.
\end{proof}

Our aim is to use the decomposition result stated in Lemma~\ref{lem:poly-chained} to prove the inclusion 
$\PolyWA \subseteq \PRat$.
It will be convenient for reasoning to use formal series.

\begin{lemma}\label{lem:chained}\hfill
\begin{itemize}
	\item The formal series induced by a chained loop of size $1$ is of the form $\frac{\alpha}{1 - \lambda x^{\ell}}$,
where $\alpha = I(q_0)$, $\lambda$ is the product of the weights in the loop and $\ell$ is the length of the loop. 
If there is no loop this reduces to $\alpha$.
	\item Let $S_1,S_2$ be the formal series induced by the chained loops $\A_1$ and $\A_2$,
	then the formal series induced by the concatenation of $\A_1$ and $\A_2$ is $x \cdot S_1 \cdot S_2$. 
	\item Let $S_1,S_2$ be the formal series induced by two automata $\A_1$ and $\A_2$,
	then the formal series induced by the union of $\A_1$ and $\A_2$ is $S_1 + S_2$.
\end{itemize}
\end{lemma}

\begin{proof}
The first and the third item are immediate, we focus on the second. For convenience let us assume that $\A_2(-1) = 0$.
By definition the concatenation of two chained loops recognises the sequence defined by 
\[
\sem{\A}(n) = \sum_{i = 0}^n \sem{\A_1}(i-1) \cdot \sem{\A_2}(n-i) 
\]
since an accepting run in the concatenation is the concatenation of an accepting run in $\A_1$ and an accepting run in $\A_2$. The only issue is that the output state of $\A_1$ was changed into a transition, and to include this step we write $\A_1(i-1)$ instead of $\A_1(i)$.
Hence the formal series is indeed the Cauchy product of $S_1$ and $S_2$, shifted by one.
\end{proof}

We are now half-way through the proof of the inclusion $\PolyWA \subseteq \PRat$:
thanks to Lemma~\ref{lem:poly-chained}, we can restrict our attention to unions of chained loops,
and thanks to Lemma~\ref{lem:chained}, we know what are the formal series induced by the sequences computed by such automata.
More specifically, they are obtained from formal series of the form $\frac{\alpha}{1-\lambda x^\ell}$
by taking sums and Cauchy products (with an additional shift).

To prove that $\PRat$ contains such sequences it is tempting to attempt showing that the sequences above are in $\PRat$
and the closure of $\PRat$ under sums and Cauchy products.
Unfortunately, the closure under Cauchy product is not clear (although it will follow from the final result that it indeed holds).

We sidestep this issue by observing that we only need to be able to do Cauchy products of formal series of a special form.
Indeed, the formal series described above are of the form $\frac{P}{Q}$ where $P,Q$ are rational polynomials 
and the roots of $Q$ are roots of rational numbers: this is true of $\frac{\alpha}{1-\lambda x^\ell}$ and is clearly closed under sums and Cauchy products (with the additional shift).

Notice that every chained loop can be obtained as concatenations of chained loops of size~$1$. Thus Lemma~\ref{lem:chained} gives a characterisation of formal series corresponding to unions of chained loops: these are sums of products of $\frac{\alpha}{1 - \lambda x^{\ell}}$ and polynomials.
We further simplify this characterisation applying the following lemma.

\begin{lemma}\label{lem:euclidian}
Consider the formal series $\frac{P}{Q}$ where $P,Q$ are rational polynomials and the roots of $Q$ are roots of rational numbers.
Then $\frac{P}{Q}$ can be written as the sum of formal series of the form $\frac{R}{(1 - \lambda x^\ell)^k}$ 
for rational polynomials $R$, rational numbers $\lambda$, and $\ell,k$ natural numbers.
\end{lemma}

\begin{proof}
This is a direct consequence of the fact that $\Q[x]$ is a Euclidean ring.
The exact statement following from this is that 
any product $\prod_{i = 1}^n \frac{R_i}{P_i}$ where the polynomials $P_i$ are mutually prime 
(meaning, for each $i$, the polynomials $P_i$ and $\prod_{j \neq i} P_j$ are coprime)
can be written as a sum of $\frac{Q_i}{P_i}$ for some rational polynomials $Q_i$.

To conclude, we observe that any polynomial whose roots are roots of rational numbers
can be written as a product of mutually prime polynomials of the form $(1 - \lambda x^\ell)^k$.
\end{proof}

By Lemma~\ref{lem:chained} and Lemma~\ref{lem:euclidian} it follows that for every finite union of chained loops its formal series is a sum of $\frac{R}{(1 - \lambda x^\ell)^k}$ for rational polynomials $R$, rational numbers $\lambda$, and $\ell,k$ natural numbers.
Combining this with Lemma~\ref{lem:poly-chained}
we get that the formal series computed by $\PolyWA$ are of the same form.
Thus we have reduced proving the inclusion $\PolyWA \subseteq \PRat$ to
proving that sequences whose formal series are sums of formal series of the form $\frac{R}{(1 - \lambda x^\ell)^k}$ are in $\PRat$. 

Since $\PRat$ is closed under sum, it suffices to consider one such formal series.
Moreover, due to the closure under shifts we can assume that the polynomial $R$ is equal to $1$; as stated in the lemma below.

\begin{lemma}\label{lem:rationa-roots}
The sequence whose formal series is $\frac{1}{(1 - \lambda x^\ell)^k}$ is in $\PRat$.
\end{lemma}

\begin{proof}
We know that
\[
\frac{1}{(1 - \lambda x^\ell)^k} = \sum_{n \in \N} \binom{n + k - 1}{k} \lambda^n x^{\ell \cdot n}.
\]
Note that $\binom{n + k - 1}{k}$ is a polynomial in $n$ of degree at most $k$, 
i.e. $\binom{n + k - 1}{k} = \sum_{p = 0}^k a_p n^p$.
It follows that
\[
\frac{1}{(1 - \lambda x^\ell)^k} = \sum_{p = 0}^k\ a_p \cdot \sum_{n \in \N} n^p \lambda^n x^{\ell \cdot n}
\]
It is enough to prove that for each $p$ the sequence whose formal series is 
\[
\sum_{n \in \N} a_p n^p \lambda^n x^{\ell \cdot n}
\]
is in $\PRat$. 
Using an arithmetic sequence and Hadamard products we construct $\langle a_p n^p \rangle_{n \in \N}$.
Multiplying it using Hadamard product with the geometric sequence $\langle \lambda^n \rangle_{n \in \N}$
yields $\langle a_p n^p \lambda^n \rangle_{n \in \N}$.
Shuffling the obtained sequence with $\ell - 1$ null sequences yields the desired sequence.
\end{proof}

\subsection{Application: the ambiguity hierarchy of weighted automata}
We show that the natural classes of weighted automata defined by ambiguity can be described using subclasses of rational expressions.

\begin{figure}[!ht]
\centering
\includegraphics[scale=.6]{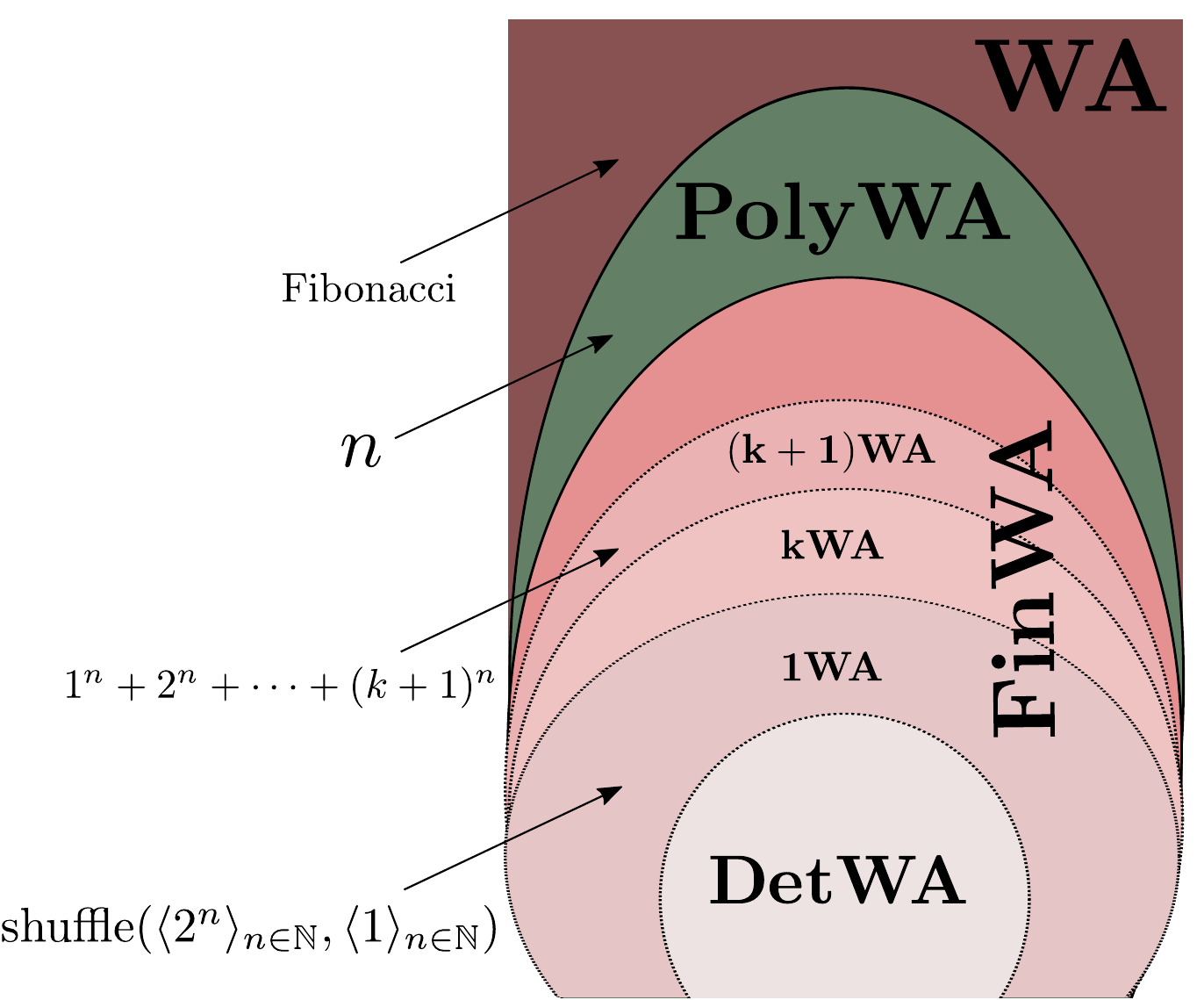}
\caption{The strict ambiguous hierarchy of weighted automata.}
\label{fig:hierarchy}
\end{figure}


\begin{lemma}\label{lem:wa-classes}\hfill
\begin{itemize}
	\item $\DetWA = \bigcup_{\lambda \in \Q} \Rat[\Geo_\lambda, \shift, \shuffle]$;
	\item $\FinWA = \Rat[\Geo, +, \shift, \shuffle]$.
\end{itemize}
\end{lemma}

\begin{proof}
We start by proving $\DetWA = \bigcup_{\lambda \in \Q} \Rat[\Geo_\lambda, \shift, \shuffle]$.

($\subseteq$) Since the automaton is deterministic it has a shape of a lasso, \ie the states can be partitioned into a path such that the last state on the path is in a loop. Let $\lambda$ be the value obtained by multiplying all values on the loop, let $l$ be the length of the loop and let $m$ be the length of the path. Then it is easy to see that the sequence is obtained by first taking a shuffle of $l$ sequences in $\Geo_\lambda$ and then shifting it $m$ times.

($\supseteq$) We already know that $\Geo_\lambda$ are definable by deterministic weighted automata from Figure~\ref{fig:arithmetic-geometric}. Closure under shift follows from the construction in the proof of $\PRat \subseteq \PolyWA$ because it preserves the property of being deterministic. The shuffle construction preserves this property only up to a certain point. The construction of each automaton $\A_i[k]$ is deterministic but taking their sum does not yield explicitly a deterministic automaton. It suffices to observe that by construction $\A_i[k]$ are all lasso automata with the same lengths of the loop. Moreover, every word is accepted by at most one $\A_i[k]$.
To define the final automaton consider $\A_i[k]$ with the longest path. The final automaton will be $\A_i[k]$ with modified transitions and final outputs. Indeed we add other automata one by one, and for every accepting state we readjust the ingoing and outgoing transitions to give the correct value.

Proof of $\FinWA = \Rat[\Geo, +, \shift, \shuffle]$.

($\subseteq$) By Lemma~\ref{lem:poly-chained} we know that each automaton in $\FinWA$ is a union of chained loops. It is easy to see that every such chained loop has to be a lasso otherwise it will contradict the assumption that the automaton is finitely ambiguous. Then the construction follows by doing the construction for every lasso as in the proof of $\DetWA = \bigcup_{\lambda \in \Q} \Rat[\Geo_\lambda, \shift, \shuffle]$ and using $+$ to deal with the union.

($\supseteq$) This follows the same steps as the proof of $\DetWA = \bigcup_{\lambda \in \Q} \Rat[\Geo_\lambda, \shift, \shuffle]$. It is even simpler because we can take a union of two automata and remain in the class of $\FinWA$.
\end{proof}

We give examples witnessing the strict inclusions $\DetWA \subsetneq \FinWA \subsetneq \PolyWA \subsetneq \WA$
and $\kWA \subsetneq \kplusoneWA$.
\begin{lemma}\label{lem:hierarchy}\hfill
\begin{itemize}
	\item $\aa = \shuffleop{\langle 2^n \rangle_{n \in \N}, \langle 1 \rangle_{n \in \N}}$ is in $\oneWA$ but not in $\DetWA$,
	\item $\uu_k$ defined by $u_n = 1^n + 2^n + \cdots + (k+1)^n$ is in $\kplusoneWA$ but not in $\kWA$,
	\item $\vv$ defined by $v_n = n$ is in $\PolyWA$ but not in $\FinWA$;
	\item Fibonacci is in $\WA$ but not in $\PolyWA$.
\end{itemize}
\end{lemma}


We omit the simple but technical proofs of the first three items.
Only the last item will be proved in Section~\ref{sec:lrs}, it follows from the fact that $\PolyWA = \PRat$ 
is equal to the class of LRS whose eigenvalues are roots of rational numbers. 
As mentioned in Example~\ref{ex:fib} the characteristic polynomial of the Fibonacci sequence is
$x^2-x-1$, so its eigenvalues are not roots of rationals.

\section{Characterisation with copyless cost-register automata}
\label{sec:ccra}
Cost-register automata (CRA)~\cite{AlurDDRY13} are deterministic automata with write-only registers, where 
each transition updates the registers using addition and multiplication. Like in Section~\ref{sec:poly} we will consider only the variant of the model over a one-letter alphabet recognising functions $f : \N \to \Q$.

Let $\X$ be a set of \emph{variables (registers)}. 
The set of \emph{expressions} $\Expr$ is generated by the following grammar
$$
e \; ::= \; x \; \mid \; r \; \mid \; e+e \; \mid \; e \cdot e,
$$
where $x \in \X$ and $r \in \Q$. 
A \emph{substitution} is a mapping $\sigma : \X \to \Expr$. 
We let $\Subs$ denote the set of all substitutions.
A \emph{valuation} is a function $\sigma : \X \to \Q$, it is a special case of substitutions, where expressions are limited to constants.
We freely compose these objects: for instance let $\X = \set{x}$, define the valuation $\nu_0(x) = 0$, the substitution 
$\sigma(x) = x+1$ and the expression $e = 2x$. 
Then $\nu_0 \circ \sigma^n \circ e = 2n$. We see this computation as the output of a $1$-register machine which initialises $x$ with $0$, 
increments its value at each step and outputs its double value.

Formally, a CRA is a tuple $\A = (Q,\X,\delta,q_0,\nu_0,\mu)$, where $Q$ is the set of states, $\X$ is the set of registers, 
$\delta : Q \to Q \times \Subs$ is the transition function, $q_0$ is the initial state, $\nu_0 : \X \to \Q$ is the initial valuation 
and $\mu : Q \to \Q$ is the final output function. 
The output of $\A$ on $n$ is defined by the unique run of length $n$:
let $q_0 \to q_1 \to \dots \to q_n$ such that $\delta(q_i) = (q_{i+1}, \sigma_{i+1})$
\[
\sem{\A}(n) = \nu_0 \circ \sigma_1 \circ \cdots \circ \sigma_n \circ \mu(q_n).
\]

A CRA is said to be linear if its transitions and output function use only linear expressions, \ie 
such that in the grammar $e \cdot e$ is restricted to $e \cdot r$.
We denote $\LCRA$ the class of sequences recognised by linear CRA, which is known to be equivalent to the class $\WA$~\cite{AlurDDRY13}.
For instance, the following linear CRA recognises the Fibonacci sequence.

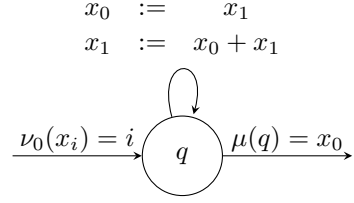
\begin{figure}[!ht]
\centering
\begin{tikzpicture}[>=stealth]
\node[state,inner sep=0pt,minimum size=30pt] (q) {$q$};

\draw (q.west) edge[<-] node[above=0.2cm,anchor=mid] {$\nu_0(x_i) = i$} ++(-1.7cm,0);
\draw (q.east) edge[->] node[above=0.2cm,anchor=mid] {$\mu(q) = x_0$} ++(1.7cm,0);
\draw (q) edge[->,loop above] node {
$
\begin{array}{ccc}
x_0 & := & x_1 \\
x_1 & := & x_0 + x_1
\end{array}
$
} (q);
\end{tikzpicture}
\caption{A linear CRA recognising the Fibonacci sequence. There is only one state and two variables $\X = \set{x_0, x_1}$. Since there is only one state the transitions are presented using only the expression that is applied every time.}\label{fig:copyless}
\end{figure}

A substitution $\sigma$ is called copyless if each register is used at most once in all $\sigma(x)$. It is easy to observe that a composition of copyless substitutions is a copyless substitution.
A CRA is said to be copyless if in each transition, each substitution is copyless. For example in Figure~\ref{fig:copyless} the register $x_1$ is used twice in the substitution so it is not a copyless automaton. 
We let $\CCRA$ denote the class of sequences recognised by copyless cost register automata (CCRA). 
In~\cite{MazowieckiR15} it is shown that $\CCRA$ is a subclass of linear CRA. We show that this is another class characterising $\PRat$.

\begin{theorem}
$\PRat = \CCRA$
\end{theorem}

\subsection*{$\PRat \subseteq \CCRA$}

This inclusion is easy to prove, it requires to perform the classical constructions as in Section~\ref{sec:poly}
and to note that they respect the copyless restriction.

\subsection*{$\CCRA \subseteq \PRat$}
%
%
We make use of a simple property in~\cite{MR18}.
A substitution is in \emph{normal form} if there exists an order on the registers $x_1 < \cdots < x_k$ 
such that the substitutions updating registers respect the order: $\sigma(x_i)$ can use only registers $x_j$ such that $x_j \ge x_i$. 
A CCRA is in normal form if all substitutions used by it are in normal form, with the same order on the registers.
It is known that every CCRA has an equivalent CCRA in normal form~\cite[Proposition 1]{MR18}.
We will use this fact only to prove Lemma~\ref{lem:normal_form}, but in the construction we will assume that the CCRA is in normal form.


Consider a CCRA $\A$, we prove that the sequence $\uu$ it recognises is in $\PRat$.
We assume without loss of generality that $\A$ is in normal form.
Since $\A$ is deterministic it has the shape of a lasso: a tail of length $k$ and a loop of length $\ell$.
Let us fix $n \in \N$ and $\ell' \in \set{0,\ldots,\ell-1}$, the run is
\begin{align}\label{eq:lasso}
q_0 \to \dots \to q_k \to \left( p_0 \to \dots \to p_{\ell - 1} \right)^n \to p_0 \to \dots \to p_{\ell'}.
\end{align}
Let $\delta(q_i) = (q_{i+1 \mod \ell},\beta_{i})$ for $i \in \set{0,\dots,k}$,
with the convention that $q_{k+1} = p_0$, 
and $\delta(p_i) = (p_{i+1 \mod \ell},\sigma_{i})$ for $i \in \set{0,\dots,\ell-1}$.
Define
\[
\nu_0' = \nu_0 \circ \beta_0 \circ \dots \circ \beta_{k}\quad ;\quad 
\sigma = \sigma_0 \circ \dots \circ \sigma_{\ell-1}\quad ;\quad 
e = \sigma_0 \circ \dots \circ \sigma_{\ell'-1} \circ \mu(p_{\ell'}).
\] 
Notice that $\sigma$ is a copyless substitution since it is a composition of copyless substitutions.
We define the sequence $\uu[\ell']$ by 
\[
u_n[\ell'] = \nu_0' \circ \sigma^n \circ e.
\]
We will prove in Lemma~\ref{lem:normal_form}  that the sequence $\uu[\ell']$ is in $\PRat$.
The decomposition of the runs into a lasso implies the following equality:
\[
\uu = \langle u_0, u_1, \ldots, u_{k-1}, \shuffleop{\uu[0],\ldots,\uu[\ell-1]} \rangle,
\]
which implies that $\uu$ is in $\PRat$, provided the lemma below is true.

\begin{lemma}\label{lem:normal_form}
For every copyless substitution $\sigma$ in normal form, for all initial valuation $\nu$
and for all expression $e$, the sequence
\[
\langle \nu \circ \sigma^n \circ e \rangle_{n \in \N}
\]
is in $\PRat$.
\end{lemma}

\begin{proof}
We prove that the sequence $\uu_x = \nu \circ \sigma^n (x)$ is in $\PRat$ for every register $x$, \ie the lemma holds for $e = x$. The general case follows since $\PRat$ is closed under addition and product. 

We consider two cases. Suppose $x$ is not used in $\sigma(x)$. We prove that for $n$ big enough the sequence stabilises, \ie $\sigma^n(x) = \sigma^{n+1}(x) = c$ for some constant $c$. We show this by the induction on the order $<$ from the assumed normal form. If $x$ is the biggest element in the order $<$ then $\sigma(x)$ is a constant and thus $\sigma^n (x) = \sigma^{n+1} (x)$. For the induction step suppose $x$ is not the biggest element. If $\sigma(x)$ is a constant then the claim is trivial. Otherwise let $x_1, \ldots ,x_m$ be registers used in $\sigma(x)$. Since $\sigma$ is copyless then $x_i$ is not used in $\sigma(x_i)$ for every $i$. Hence by the induction assumption for every $i$ there exists $n_i$ such that $\sigma^{n}(x_i) = \sigma^{n+1}(x_i)$ for all $n \ge n_i$. It suffices to take $n = \max_i\{n_i \mid 1 \le i \le m\} + 1$. Since constant sequences are geometric sequences with $\lambda = 1$ then $\uu_x$ can be defined in $\PRat$ using shift.

Now suppose that $x$ is used in $\sigma(x)$. The expression $\sigma(x)$ is equivalent to $\sum_{i=0}^m a_i \cdot x_i$ for some constants $a_i$, where $x_0 = x$ and $x_i$ are pairwise different. Since $\sigma$ is copyless then for all $i > 0$ we know that $\sigma(x_i)$ does not use $x_i$. By the previous paragraph there exists $N$ such that $\sigma^N(x_i) = \sigma^{N+1}(x_i) = c_i$ for some constants $c_i$ for all $i > 0$. Let $n \ge N$. Then
\[
\nu \circ \sigma^{n+1}(x) \; = \; \nu \circ \sigma^n \circ \sigma(x) \; = \; \nu \circ \left( \sum_{i=0}^m a_i \cdot\sigma^n(x_i) \right) \; = \; a_0 \cdot \left( \nu \circ \sigma^n(x) \right) + \sum_{i = 1}^m a_i \cdot c_i. 
\]
Let $a = a_0$ and $b = \sum_{i = 1}^m a_i \cdot c_i$. We proved that for $n \ge N$ the sequence $\uu_x$ satisfies $u_x(n+1) = a \cdot u_x(n) + b$. It remains to prove that this sequence is in $\PRat$. It is enough to show that $\uu_x'(n) = \uu_x(n+N)$ is in $\PRat$ since to obtain $\uu_x$ it suffices to use shift $N$ times. There are two cases. If $a = 1$ then $\uu_x'(n)$ is an arithmetic sequence, which concludes the proof. If $a \neq 1$ then
\[
u_x'(n) \; = \; a^n \cdot u_x'(0) + \sum_{i = 0}^{n-1} a^i \cdot b \; = \; a^n \cdot u_x'(0) + b \cdot \frac{a^n - 1}{a - 1}.
\]
This is a sum of a geometric sequence $a^n \cdot (u_x'(0) + \frac{b}{a-1})$; and a constant sequence $- \frac{b}{a-1}$; which proves $\uu_x'$ is in $\PRat$.
\end{proof}

\begin{remark}
One can extract from this proof the equivalence between linear $\CCRA$ and $\Rat[\Arith \cup \Geo,+,\shift,\shuffle]$.
\end{remark}

It was recently shown that $\CCRA$ are strictly less expressive than weighted automata~\cite{MR18}. 
The proof goes by analysing the Fibonacci sequence. 
We will get as a corollary of our results a self-contained proof that $\LCRA$ and $\CCRA$ are different. 

\section{Characterisation with linear recurrence sequences and formal series}
\label{sec:lrs}
Our last two characterisations are as follows.

\begin{theorem}\label{th:eigenvalues}
$\PRat$ is the class of LRS whose eigenvalues are roots of rational numbers,
and equivalently whose formal series are $\frac{P}{Q}$ with $P,Q$ rational polynomials and the roots of $Q$ are roots of rational numbers.
\end{theorem}

Before proving the theorem, we note that we can now substantiate the claim that 
the Fibonacci sequence is not in $\PRat$ (hence not in $\CCRA$ and $\PolyWA$),
since its eigenvalues are not roots of rational numbers.

We rely on the following classical result about LRS, see e.g.~\cite{0076724}.

\begin{lemma}\label{lem:classical_LRS}
Let $\uu$ be an LRS and $Q$ its characteristic polynomial.
The formal series induced by $\uu$ is $\frac{P}{Q}$ for some rational polynomial $P$.
\end{lemma}

For both inclusions we rely on Theorem~\ref{th:poly} stating that $\PRat = \PolyWA$
and the decompositions obtained in the subsequent lemmas.

\subsection*{$\PRat \subseteq$ LRS whose eigenvalues are roots of rational numbers}

By Lemma~\ref{lem:poly-chained} and Lemma~\ref{lem:chained} the formal series of sequences in $\PolyWA$ 
are sums and Cauchy products of formal series of the form $\frac{R}{1 - \lambda x^\ell}$, where $R$ is a rational polynomial, $\ell \in \N$ and $\lambda \in \Q$. 
The roots of $1 - \lambda x^\ell$ are roots of $\frac{1}{\lambda}$, so the roots of the characteristic polynomial are roots of rational numbers.

\subsection*{LRS whose eigenvalues are roots of rational numbers $\subseteq \PRat$}

Consider an LRS whose eigenvalues are roots of rational numbers.
Thanks to Lemma~\ref{lem:classical_LRS} the formal series it induces is $\frac{P}{Q}$ with $P,Q$ rational polynomials 
and the roots of $Q$ are roots of rational numbers. 
By Lemma~\ref{lem:euclidian} the formal series can be written as a sum of formal series of the form $\frac{R}{(1 - \lambda x^\ell)^k}$ 
for rational polynomials $R$, rational number $\lambda$, and $\ell,k$ natural numbers. 
It follows from Lemma~\ref{lem:rationa-roots} and the closure of $\PRat$ under sum and shift that such sequences belong to $\PRat$.

\section{Conclusion}
\label{sec:conclusion}
We introduced a class of linear recurrence sequences and obtained several characterisations.
The most surprising equivalence is $\CCRA = \PolyWA$.
This equality is very particular to our setting: for instance the two classes are incomparable, \ie neither of the inclusions hold, for tropical semirings~\cite{MR18,MazowieckiR18}.
We also conjecture that these classes are incomparable over the rational semiring for general alphabets (of size bigger than~1).

We leave open the precise complexity of the Skolem problem for $\PRat$.
Recent progress has been made for a subclass of $\PRat$~\cite{ABV17}: 
the Skolem problem for LRS whose eigenvalues are roots of unity is NP-complete.
Our class is more general since we consider LRS whose eigenvalues are roots of rational numbers,
so the NP-hardness also applies. However the algorithm constructed in~\cite{ABV17} does not extend to our class.



\bibliography{biblio}

\end{document}